\documentclass[12pt]{amsart}
\usepackage[dvipdfmx]{graphicx}

\usepackage{amssymb, enumerate, mdwlist}

\evensidemargin\paperwidth
\advance\evensidemargin-\textwidth
\advance\oddsidemargin-1in
\evensidemargin\oddsidemargin

\topmargin\paperheight
\advance\topmargin-\textheight
\advance\topmargin-1in

\newcommand{\B}{\mathcal{B}}
\newcommand{\C}{\mathbb{C}}
\newcommand{\ee}{\mathbf{e}}
\newcommand{\F}{\mathcal{F}}
\newcommand{\Hil}{\mathcal{H}}
\newcommand{\ii}{\mathbf{i}}
\newcommand{\kk}{\mathbf{k}}
\newcommand{\N}{\mathbb{N}}

\newcommand{\R}{\mathbb{R}}
\newcommand{\T}{\mathbb{T}}
\newcommand{\vv}{\mathbf{v}}
\newcommand{\ww}{\mathbf{w}}
\newcommand{\xx}{\mathbf{x}}
\newcommand{\Z}{\mathbb{Z}}

\theoremstyle{plain}
\newtheorem{theorem}{Theorem}[section]
\theoremstyle{plain}
\newtheorem{proposition}[theorem]{Proposition}
\theoremstyle{plain}
\newtheorem{lemma}[theorem]{Lemma}
\theoremstyle{plain}
\newtheorem{corollary}[theorem]{Corollary}
\theoremstyle{plain}

\theoremstyle{plain}
\newtheorem{definition}[theorem]{Definition}
\theoremstyle{plain}

\theoremstyle{remark}
\newtheorem{remark}[theorem]{Remark}
\theoremstyle{remark}
\newtheorem{example}[theorem]{Example}
\theoremstyle{remark}

\title
[multi-dimensional QWs]
{Convergence theorems on multi-dimensional homogeneous quantum walks}

\author{Hiroki Sako}
\address
{Faculty of Engineering, Niigata University, Nishi-ku, Niigata 950-2181, Japan}
\email
{sako@eng.niigata-u.ac.jp}

\begin{document}
\maketitle

\begin{abstract}
We propose a general framework for quantum walks on $d$-dimensional spaces.
We investigate asymptotic behavior of these walks.
We prove that every homogeneous walks with finite degree of freedom has limit distribution.
This theorem can also be applied to every crystal lattice.
In this theorem, it is not necessary to assume that the support of the initial unit vector is finite.
We also pay attention on $1$-cocycles,
which is related to Heisenberg representation of time evolution of observables.
For homogeneous walks with finite degree of freedom, convergence of averages of $1$-cocycles associated to the position observable is also proved.
\end{abstract}

\section{Introduction}
A quantum walk is a dynamical system given by a unitary operator $U$ on a Hilbert space $\Hil$.
The space $\Hil$ is often associated to some metric spaces like $\Z^d$ and $\R^d$. 
The first purpose of this paper is to propose a general framework for quantum walks on $d$-dimensional spaces. See Subsection \ref{subsection: definition}.
This new framework includes the quantum walks studied in \cite{MBSS}, \cite{Tregenna}, \cite{GJS}, \cite{InuiKonishiKonno}, \cite{KonnoJMSJ}, \cite{OliveiraPortugalDonangelo}, \cite{WatabeKobayashiKatoriKonno}, and in many other papers.

The second purpose is to observe asymptotic behavior of such walks.
For an initial unit vector $\xi$ in $\Hil$, we obtain a sequence $\{U^t \xi \}_{t \in \N}$ of unit vectors.
It gives a sequence of probability measures $\{p_t\}_{t \in \N}$, which is related to the probability interpretation of quantum mechanics.
It is proved in Theorem \ref{theorem: asymptotic concentration} that the probability measures are asymptotically concentrated on some small region.
This theorem can be applied to walks which is not necessarily space-homogeneous.
In this theorem, we need some very weak condition on the walk $U$ called smoothness defined  in Subsection \ref{subsection: regularity}.

The third purpose is to show existence of limit distributions for analytic homogeneous quantum walks with respect to an arbitrary initial unit vector in Theorem \ref{theorem: limit distribution}.
Analyticity is so mild condition that almost all the known quantum walks satisfy it.
For the proof, we study $1$-cocycles of such a walk in Subsection \ref{subsection: logarithmic derivatives}.
$1$-cocycles are related to Heisenberg representation of time evolution of observables.
Let $U$ be a quantum walk acting on a Hilbert space $\Hil$, and let $D$ be some observable of position.
The sequence of observables $\{U^{-t} D U^t\}_t$ stands for the time evolution of the observable $D$.
The sequence $\{c_t = U^{-t} D U^t - D\}_t$ is the most important example of $1$-cocycle and is called a logarithmic derivative.

For more than 20 years, the most important 
subject of quantum walks is the following form of unitary operator:
\begin{eqnarray}\label{equation: concrete QW}
U = 
\left(
\begin{array}{cc}
S & 0\\
0 & S^{-1}
\end{array}
\right)
\left(
\begin{array}{cc}
\alpha & - \overline{\beta}\\
\beta & \overline{\alpha}
\end{array}
\right)
\end{eqnarray}
and their higher dimensional version.
Here, $\alpha$ and $\beta$ are complex numbers satisfying 
$|\alpha|^2 + |\beta|^2 = 1$, and the operator $S$ is the bilateral shift on the Hilbert space $\ell_2 (\Z)$ consisting of the square summable functions on the integer group $\Z$.
The operator $U$ acts on the Hilbert space $\ell_2(\Z) \otimes \C^2$.
It has been already shown that several kinds of space-homogeneous quantum walks have limit distributions $\lim_{t \to \infty} p_t$.
Although their argument was not rigorous,
Grimmett, Janson, and Scudo presented an excellent idea to show such a convergence theorem in \cite{GJS}.
This paper improves their result in the following aspects:
\begin{itemize}
\item
Our paper clarify what kind of property of the operator $U$ really works in the proof of the convergence theorem.
It turns out that analyticity of the operator works.
\item
We no longer need smooth eigenvalue functions of the inverse Fourier transform of the quantum walk.
\footnote{In fact, it is possible to construct a quantum walk whose eigenvalue function is not smooth.
Because we need many pages for the construction, we omit.
Our convergence theorem also works for such an example.}
\item
Our theorem can be applied to many kinds 
of homogeneous quantum walks, because
analyticity for the quantum walks is a weak condition.
For example, a quantum walks on arbitrary crystal lattice with translation symmetry with finite propagation satisfies this condition.
\item
Our theorem does not require locality of its initial unit vector.
\end{itemize}


\section{Definition of quantum walks and their regularity}

\subsection{Definition of multi-dimensional QWs}
\label{subsection: definition}
Throughout this paper, $\B(\Hil)$ stands for the set of all the bounded linear operators on a Hilbert space $\Hil$.

\begin{definition}\label{definition: QWs}
Let $d$ be a natural number.
A triple $\left( \Hil, \left(U^t \right)_{t \in \Z}, E \right)$ is said to be 
{\rm a $d$-dimensional quantum walk}, 
if the following conditions holds:
\begin{enumerate}
\item
$\Hil$ is a Hilbert space,
\item
$\left( U^t \right)_{t \in \Z}$ is a unitary representation of $\Z$ on $\Hil$.
\item\label{item: self-adjoint operators}
$E$ is a Borel measure on $\R^d$ whose values are orthogonal projections in $\B(\Hil)$ such that $E(\R^d) = \mathrm{id}_{\Hil}$.
\end{enumerate}
The measure $E$ in the last item is called {\rm a spectral measure}.
\end{definition}

In many references, researchers concentrate on quantum walks on the Hilbert space $\ell_2(\Z^d) \otimes \C^d$.

\begin{example}\label{example: standard spectral measure}
For the Hilbert space $\ell_2(\Z^d) \otimes \C^d$,
the measure $E$ is defined as follows:
for every subset $\Omega \subset \R^d$, the operator $E(\Omega)$ is the orthogonal projection from $\ell_2(\Z^d) \otimes \C^d$
onto $\ell_2(\Omega \cap \Z^d) \otimes \C^d$.
A unitary operator $U$ on $\ell_2(\Z^d) \otimes \C^d$
defines a triplet $\left( \Hil, \left(U^t \right)_{t \in \Z}, E \right)$,
which we call {\it a quantum walk} in this paper.
\qed
\end{example}

Throughout this paper, we often make use of the following items.
\begin{itemize}
\item
For $i \in \{1, \cdots, d\}$, a self-adjoint operator $D_i$ on $\Hil$ is defined by
\[D_i = \int_{(x_i) \in \R^d} x_i E(d x_1 \cdots d x_i \cdots d x_d).\]
\item
We denote by $( \cdot , \cdot )$ the standard real-valued inner product on $\R^d$.
The projection-valued measure $E$ gives a unitary representation of $\R^d$:
\[\sigma \colon \kk \mapsto \int_{\xx \in \R^d} \exp(\ii (\kk, \xx)) E(d\xx).\]
This representation is given by 
$\sigma((k_i)_i) = \prod_{i = 1}^d \exp(\ii k_i D_i)$.
The representation $\sigma$ is continuous with respect to the strong operator topology.
Throughout this paper, $\ii$ stands for the imaginary unit,
and $i$ is a natural number.
\item
We denote by $\alpha_\kk$ 
the $d$-parameter automorphism group on $\B(\Hil)$ defined by
\[\alpha_\kk(X) = \sigma(\kk) X \sigma(- \kk).\]
\end{itemize}

In many cases, $D_i$ stands for the position observable.
However, Definition \ref {definition: QWs} allows us to treat quantum walks in more flexible manner.
The self-adjoint operator $D_i$ can be other quantum mechanical observables like momentum.

\subsection{Regularity on bounded operators}
\label{subsection: regularity}

\begin{definition}
Let $X$ be a bounded operator on $\Hil$.
The operator $X$ is said to be {\rm uniform} with respect to $E$, 
if the mapping $\R^d \ni \kk \mapsto \alpha_\kk(X) \in \B(\Hil)$ is continuous with respect to the operator norm on $\B(\Hil)$.
\end{definition}
If the operator $X$ is uniform, then the mapping $\kk \mapsto \alpha_\kk(X)$  is uniformly continuous.
It is easy to see that the set of uniform operators is a closed subset 
$\B(\Hil)$.
Using the operator norm on $\B(\Hil)$, we may consider partial derivatives.

\begin{definition}
For $i \in \{1, \cdots, d\}$, define the {\rm partial derivative} $\partial_i (X)$ by the norm limit
\begin{eqnarray}\label{equation: partial derivative}
\partial_i (X) = \lim_{k \to 0} \frac{\exp(\ii k D_i) X \exp(- \ii k D_i) - X}{k},
\end{eqnarray}
if it exists.
\end{definition}

\begin{example}\label{example: bilateral shift}
Let $\ell_2(\Z^d)$ be the Hilbert space of all the square summable functions on $\Z^d$.
For $\xx \in \Z^d$, denote by $\delta_\xx$ the
definition function of $\{\xx\} \subset \Z^d$.
Let $E$ be the standard spectral measure on $\R^d$ defined in Example
\ref{example: standard spectral measure}. 
Let $j$ be an element of $\{1, \cdots, d\}$.
Define the unitary operator $S_j$ by
$\delta_{\xx} \mapsto \delta_{\xx + \ee_j}$,
where $\{\ee_1, \cdots, \ee_j, \cdots, \ee_d\}$ stands for the standard basis of $\Z^d$.

By the straightforward calculation,
for every $\kk = (k_1, \cdots, k_j, \cdots, k_d) \in \R^d$,
we have 
\begin{eqnarray*}\label{equation: the action of kk}
\alpha_\kk(S_j) = \exp(\ii k_j) S_j.
\end{eqnarray*}
Since $\exp(\ii k_j)$ is a continuous function on $\R^d$, it turns out that $S_j$ is uniform.

By equation (\ref{equation: partial derivative}),
for $i \neq j$,
we have
\begin{eqnarray*}
\partial_i (S_j) 
= \lim_{k \to 0} \frac{\exp(\ii k D_i) S_j \exp(- \ii k D_i) - S_j}{k} = 0.
\end{eqnarray*}
We also have
\begin{eqnarray*}
\partial_j (S_j) 
&=& \lim_{k_j \to 0} \frac{\exp(\ii k_j D_i) S_j \exp(- \ii k_j D_i) - S_j}{k_j} \\
&=& \lim_{k_j \to 0} \frac{\exp(\ii k_j) - 1}{k_j} S_j\\
&=& \ii S_j.
\end{eqnarray*}
We can also calculate the higher derivatives as above.
It follows that the unitary operator $S_j$ is smooth.

In particular, the bilateral shift $S$ on $\ell_2(\Z)$ is smooth. 
\qed
\end{example}

\begin{example}
Consider the case that $d = 1$.
Let $U$ be the most famous $1$-dimensional quantum walk 
presented in equation (\ref{equation: concrete QW}) 
acting on $\ell_2(\Z) \otimes \C^2$.
Let $E$ be the standard spectral measure on $\ell_2(\Z) \otimes \C^2$.
Since the bilateral shift $S$ on $\ell_2(\Z)$ is smooth.
$U$ is smooth with respect to $E$.
\qed
\end{example}

\begin{lemma}\label{lemma: partial derivative and commutator}
If $\partial_i (X)$ exists, then
$X (\mathrm{dom}(D_i)) \subset \mathrm{dom}(D_i)$, and $\partial_i (X)$ 
is the unique extension of the commutator
$[\ii D_i, X] = \ii (D_i X - X D_i)$
defined on the domain $\mathrm{dom}(D_i)$.
For every $\kk \in \R^d$,
$\partial_i (\alpha_\kk(X))$ exists and is equal to $\alpha_\kk( \partial_i (X))$.
\end{lemma}

\begin{proof}
The first assertion is well-known.
See \cite[Lemma 2.4]{Sako} for example.
We calculate $\partial_i (\alpha_\kk(X))$ as follows:
\begin{eqnarray*}
&&
\lim_{k \to 0} \frac{\exp(\ii k D_i) \sigma(\kk) X \sigma(- \kk) \exp(- \ii k D_i) - \sigma(\kk) X \sigma(- \kk)}{k}\\
&=&
\sigma(\kk) \lim_{k \to 0} \frac{\exp(\ii k D_i) X \exp(- \ii k D_i) - X}{k} \sigma(- \kk)
=
\alpha_\kk (\partial_i(X)).
\end{eqnarray*}
\end{proof}

\begin{definition}
Let $X$ be an operator on $\Hil$ which is uniform with respect to $E$.
The operator $X$ is said to be {\rm smooth} with respect to $E$, 
if for every sequence $\{i(m)\}$ of $\{1, \cdots, d\}$, the higher order partial derivatives
\[\partial_{i(1)}(X), \quad
\partial_{i(2)} ( \partial_{i(1)}(X)), \quad
\partial_{i(3)} ( \partial_{i(2)} ( \partial_{i(1)}(X))), \quad
\cdots \]
exist and are uniform with respect to $E$.
\end{definition}

\begin{lemma}
\label{lemma: adjoint and multiplication}
Suppose that $X, Y \in \B(\Hil)$ are smooth with respect to $E$.
Then $X^*$ and $XY$ are also smooth and satisfy
\[\partial_i(X^*) = \partial_i(X)^*, \quad \partial_i(X Y) = \partial_i(X) Y + X \partial_i(Y)\]
for every $i \in \{1, \cdots, d\}$.
\end{lemma}

\begin{proof}
By equation (\ref{equation: partial derivative}),
we have the following equation of limits in norm topology:
\begin{eqnarray*}
\partial_i (X)^* 
&=& \left(\lim_{k \to 0} 
\frac{\exp(\ii k D_i) X \exp(- \ii k D_i) - X}{k} \right)^*\\
&=& \lim_{k \to 0} 
\left(\frac{\exp(\ii k D_i) X \exp(- \ii k D_i) - X}{k} \right)^*\\
&=& \lim_{k \to 0} \frac{\exp(\ii k D_i) X^* \exp(- \ii k D_i) - X^*}{k}\\
&=& \partial_i (X^*).
\end{eqnarray*}
Therefore every partial derivatives of $X^*$ exist.
Since $\partial_i (X)$ is uniform, its adjoint is also uniform.
Repeating this calculation of partial derivatives,
every higher derivatives of $X^*$ also exist and are uniform.
This means that $X^*$ is also smooth.

Since $Y$ is uniform, the equation
\[Y = \lim_{k \to 0} \exp(\ii k D_i) Y \exp(- \ii k D_i)\]
holds.
Combining with equation (\ref{equation: partial derivative}) for $X$ and for $Y$,
we have
\begin{eqnarray*}
&&
\partial_i(X) Y + X \partial_i(Y)
=
\partial_i(X) 
\left(\lim_{k \to 0} \exp(\ii k D_i) Y \exp(- \ii k D_i) \right) 
+ X \partial_i(Y)
\\
 & = &
\lim_{k \to 0} 
\frac{\exp(\ii k D_i) X  Y \exp(- \ii k D_i) - X\exp(\ii k D_i) Y \exp(- \ii k D_i)}{k} \\
&&+
\lim_{k \to 0} X \frac{\exp(\ii k D_i) Y \exp(- \ii k D_i) - Y}{k}\\
&=&
\lim_{k \to 0} 
\frac{\exp(\ii k D_i) X  Y \exp(- \ii k D_i) - XY}{k}\\
&=&
\partial_i(X Y).
\end{eqnarray*}
Therefore every partial derivatives of $XY$ exist.
Since $\partial_i(X), Y, X, \partial_i(Y)$ are uniform, 
$\partial_i(X Y)$ is also uniform.
Repeating this calculation of partial derivatives,
every higher derivatives of $X Y$ also exist and are uniform.
This means that $XY$ is also smooth.
\end{proof}

\begin{definition}\label{definition: analyticity}
Let $X$ be a bounded operator on $\Hil$.
The operator $X$ is said to be {\rm analytic} with respect to $E$, 
if the mapping $\kk \mapsto \alpha_\kk(X)$
can be extended to a holomorphic mapping defined on a neighborhood
\[\{(\kappa_1, \cdots, \kappa_d) \in \C^d \ |\ - \delta < \mathrm{Im}(\kappa_i) < \delta \mathrm{\ for\ } 1 \le i \le d\}\]
of $\R^d$,
where $\delta$ is some positive number.
\end{definition}

Note that every analytic operator is smooth.

\begin{definition}
A $d$-dimensional quantum walk $\left( \Hil, \left(U^t \right)_{t \in \Z}, E \right)$ is said to be {\rm uniform} {\rm (smooth, analytic)}, if $U$ is uniform (smooth, analytic) with respect to $E$.
\end{definition}

We note that analyticity on quantum walks is so weak that the class contains all the multi-dimensional quantum walks which have been studied.

\begin{example}
Let $j$ be in $\{1, \cdots, d\}$.
Let $S_j$ be the unitary operator on $\ell_2(\Z^d)$ given by the shift as in Example \ref{example: bilateral shift}.
The action $\alpha_\kk(S_j)$ of $\kk = (k_1, \cdots, k_d) \in \R^d$
is given by equation (\ref{equation: the action of kk}),
$\alpha_\kk(S_j) = \exp(\ii k_j) S_j$.
This action can be extended to $\kk = (k_1, \cdots, k_d) \in \C^d$ by $\kk \mapsto \exp(\ii k_j) S_j$. This map is holomorphic.
It follows that $S_j$ is analytic.
In particular the bilateral shift $S$ on $\ell(\Z)$ is analytic.
Therefore,
the most famous $1$-dimensional quantum walk 
presented in equation (\ref{equation: concrete QW}) 
is also analytic.
\end{example}

\subsection{Similarity between multi-dimensional QWs}
\label{subsection: similarity}

Let $\Hil_1$ and $\Hil_2$ be Hilbert spaces.
For $i = 1, 2$,
let $E_i$ be a Borel measure on $\R^d$ whose values are orthogonal projections in $\B(\Hil_i)$.
Define a new projection-valued measure $E_1 \oplus E_2$ by
\[E_1 \oplus E_2(\Omega) = E_1(\Omega) \oplus E_2(\Omega), \quad {\rm \ a \ Borel \ subset \ } \Omega \subset \R^d\]
A bounded operator $X \colon \Hil_1 \to \Hil_2$ is said to be {\it smooth} with respect to $E_1$ and $E_2$, if 
$
\left(
\begin{array}{cc}
0 & 0\\
X & 0
\end{array}
\right)
\colon  \Hil_1 \oplus \Hil_2 \to \Hil_1 \oplus \Hil_2
$
is smooth with respect to $E_1 \oplus E_2$.

\begin{definition}\label{definition: similarity}
Two $d$-dimensional quantum walks $(\Hil_1, (U_1^t)_{t \in \Z}, E_1)$, $(\Hil_2, (U_2^t)_{t \in \Z}, E_2)$ are said to be {\rm similar},
if there exists a unitary operator $V \colon \Hil_1 \to \Hil_2$ satisfying the following 
conditions
\footnote{In this definition, we require that $V$ is unitary.
This condition can be relaxed.
It is possible replace the smooth unitary intertwiner $V$ with a smooth {\it invertible} intertwiner.
These definitions are equivalent.
Because its proof is long, we omit the explanation.}
:
\begin{itemize}
\item
$V U_1 = U_2 V$, 
\item 
$V$ is smooth with respect to
$E_1$ and $E_2$.
\end{itemize}
\end{definition}

Since $V^{-1}$ is equal to $V^*$,
by Lemma \ref{lemma: adjoint and multiplication},
similarity between quantum walks is a reflexive relation.
To see similarity between quantum walks is transitive,
assume 
that $V_1$ is a smooth unitary intertwiner from a quantum walk $U_1$ to a quantum walk $U_2$, 
and that $V_2$ is a smooth unitary intertwiner from a $U_2$ to a quantum walk $U_3$.
Since we have 
\[V_2 V_1 U_1 = V_2 U_2 V_1 = U_3 V_2 V_1,\]
$V_2 V_1$ is a unitary intertwiner from $U_1$ to $U_3$.
Again by Lemma \ref{lemma: adjoint and multiplication},
the unitary $V_2 V_1$ is also smooth.
It follows that $U_1$ and $U_3$ are similar.
Therefore, similarity is a equivalence relation.
If two quantum walks are similar and if one of the walks is smooth, then the other walk is also smooth.

\begin{remark}
As in Proposition \ref{proposition: similar asymptotic behavior},
if two quantum walks are similar,
and if one of them has a limit distribution,
then the other walk also has a limit distribution.
Then the limit is the same as that of the other.
For the study of limit distributions of a quantum walk $U$,
we can replace it with another walk which is similar to $U$.
For a Hilbert space $\Hil$ with a $d$-dimensional coordinate system $E$,
we can modify the information of position in the following sense.
\end{remark}

\begin{example}
Let $X \subset \R^d$ is a discrete subset.
Consider the case that for every $x \in X_1$, a Hilbert space $\Hil_x$ is given.
Define a Hilbert space $\Hil_X$ by $\oplus_{x \in X} \Hil_x$.
Let $E_X$ be the spectral measure on $\R^d$ defined by
the following:
for $\Omega \subset \R^d$,
$E_X(\Omega)$ is the orthogonal projection from $\Hil_X$ to $\oplus_{x \in \Omega \cap X} \Hil_x$.

Let $f \colon X \to \R^n$ be a map satisfying that there exist a constant $0 < R$ such that
for every $x \in X$, the distance between $x$ and $f(x)$ is at most $R$.
The map $f$ stands for the modification of position.
Define $Y$ by the image $f(X)$.
For $y \in Y$, define $\mathcal{K}_y$ by $\oplus_{x \in f^{-1}(y)} \Hil_x$.
Define $\mathcal{K}_Y$ by $\oplus_{y \in Y} \mathcal{K}_y$.
Let $E_Y$ be the spectral measure on $\R^d$ defined by
the following:
for $\Omega \subset \R^d$,
$E_Y(\Omega)$ is the orthogonal projection 
from $\mathcal{K}_Y$ to $\oplus_{y \in \Omega \cap Y} \mathcal{K}_y$.

Define a unitary $V$ from $\Hil_X = \oplus_{x \in X} \Hil_x$ to
$\mathcal{K}_Y = \oplus_{y \in Y} \mathcal{K}_y 
= \oplus_{y \in Y} \oplus_{x \in f^{-1}(y)} \Hil_x$
by the direct sum of $\mathrm{id}_{\Hil_x} \colon \Hil_x \to \Hil_x$.
Then the unitary $V$ is smooth or more strongly analytic with respect to
$E_X$ and $E_Y$.

Let $U$ be a quantum walk acting on $(\Hil_X, E_X)$.
Then we also have a walk $V U V^{-1}$ acting on $(\mathcal{K}_Y, E_Y)$.
Since $V$ is analytic (therefore smooth and uniform), the walk $V U V^{-1}$ is analytic (smooth, or uniform), 
if and only if $U$ has the same regularity.
Proposition \ref{proposition: similar asymptotic behavior}
shows that asymptotic behavior of $V U V^{-1}$ is the same as $U$.

We can apply this example to every crystal lattice.
Let $d$ be $2$ or $3$.
Assume that $X \subset \R^{d}$ forms a crystal lattice.
More precisely, there exists an additive subgroup $G \subset \R^d$ which is isomorphic to $\Z^d$ such that for every $g \in G$, $g + X = X$.
We can chose some bounded fundamental domain $\Xi \subset \R^d$ 
of the additive action of $G$ on $\R^d$,
that is, the family $\{g + \Xi\}_{g \in G}$ are disjoint and 
its union is $\R^d$.
The intersection of $\Xi$ and $X$ stands for the unit of the crystal structure.
Choose a point $y_0$ in $\Xi$.
Define $Y \subset \R^d$ by $\{ g + y_0 \ | \ g \in G\}$.
Define $f \colon X \to Y$ as follows:
for every $x \in (g + \Xi) \cap X$, $f(x) = g + y_0$.
This map $f$ stands for the modification of position.

For a quantum walk $U$ on a Hilbert space associated to $X$,
we can consider a quantum walk $V U V^{-1}$ on a Hilbert space associated to $Y$ which is similar to $U$.
In this way, we reduce the study of quantum walks 
on arbitrary crystal lattice to those on the integer lattice $(G + y_0) \cong \Z^d$.
\qed
\end{example}

\section{General theory for asymptotic behavior of quantum walks}
\label{section: general theory of asymptotic behavior}

\subsection{Logarithmic derivatives and their asymptotic behavior}
\label{subsection: logarithmic derivatives}

\begin{definition}\label{definition: 1-cocycle}
Let $\left( U^t \right)_{t \in \Z}$ be a unitary representation of $\Z$ on a Hilbert space $\Hil$.
A two-sided sequence $\{c_t\}_{t \in \Z}$ of bounded operators on $\Hil$ is called a $1$-cocycle of $\left( U^t \right)_{t \in \Z}$, if for every $s, t \in \Z$,
$c_{s + t} = U^{-t} c_s U^t + c_t$.
\end{definition}

\begin{lemma}\label{lemma: bounded by c1}
If $\{c_t\}_{t \in \Z}$ is a $1$-cocycle of $\left( U^t \right)_{t \in \Z}$, then for every $t \ge 1$, $\| c_t \| \le t \|c_1\|$.
\end{lemma}

\begin{proof}
For every $s, t \in \Z$,
we have
\begin{eqnarray}
\label{equation: subadditive}
\|c_{s + t} \| \le \|U^{-t} c_s U^t\| + \|c_t\| = \|c_s\| + \|c_t\|.
\end{eqnarray}
Repeating this decomposition, we obtain $\|c_t\| \le t \|c_1\|$.
\end{proof}

\begin{lemma}
The limit $\lim_{t \to \infty} \frac{\|c_t\|}{t}$ exists and is less than $\| c_1 \|$.
\end{lemma}

\begin{proof}
By the inequality (\ref{equation: subadditive}), the sequence $\{\|c_t\|\}$ is subadditive.
For every subadditive sequence $\{\gamma_t\}_{t = 1}^\infty$ of real numbers, the sequence $\left\{ \gamma_t / t \right\}$ converges to its infimum.
\end{proof}

For the rest of this paper, the $1$-cocycle associated to the observable of position plays a key role.
Let $\left( \Hil, \left(U^t \right)_{t \in \Z}, E \right)$ be a $d$-dimensional smooth quantum walk.
Let $\ww = (w_i)$ be a vector in $\R^d$.
Denote by $D$ the self-adjoint operator
\[\sum_{i = 1}^d w_i D_i = \sum_{i = 1}^d \int_{(x_i) \in \R^d} w_i x_i E(d x_1 \cdots d x_d).\]
Since $U^t$ is smooth with respect to $E$, 
by
Lemma \ref{lemma: partial derivative and commutator},
the commutator of $U^t$ with $\ii D_i$ is bounded.
It follows that the operator
\begin{eqnarray}\label{equation: coboundary of position}
U^{-t} D U^t - D = - \ii \sum_{i = 1}^d U^{-t} w_i [\ii D_i, U^t]
\end{eqnarray}
uniquely defines a bounded operator $c_t$ on $\Hil$.
We note that the sequence $U^{-t} D U^t$ stands for the time evolution of the observable $D$ in the Heisenberg representation.
The two-sided sequence $\{c_t\}_{t \in \Z}$ is a $1$-cocycle of $\left( U^t \right)_{t \in \Z}$.
By
Lemma \ref{lemma: partial derivative and commutator}, and by equation
(\ref{equation: coboundary of position}),
the $1$-cocycle is equal to
\footnote{The equation is similar to the right hand side of the following formula in calculus: 
$(\log f(t))' = f(t)^{-1} f'(t)$.
This is the motivation of the definition of {\it the logarithmic derivative}.}
\begin{eqnarray}\label{equation: logarithmic derivative of a QW}
c_t = - \ii \sum_i w_i U^{-t} \partial_i(U^t).
\end{eqnarray}

\begin{definition}
The $1$-cocycle
$\left\{c_t = U^{-t} D U^t - D \right\}$
is called {\rm the logarithmic derivatives} of the quantum walk $\left( \Hil, \left(U^t \right)_{t \in \Z}, E \right)$ with respect to the operator $D = \sum_i w_i D_i$.
The operator
$\left\{c_t / t \right\}$
is called {\rm the average of logarithmic derivatives} of the quantum walk.
\end{definition}

In most cases, the norms of the logarithmic derivatives $\left\{\|c_t\| \right\}$ linearly increases.
Asymptotic behavior of
the average $\left\{c_t / t \right\}$ of the logarithmic derivatives is an important tool in the study of limit distribution of the quantum walk (Theorem \ref{theorem: limit distribution}).
The following proposition means that 
if two walks are similar, then averages of logarithmic derivatives are similar.

\begin{proposition} \label{proposition: similarity and ALD}
Let $(\Hil_1, (U_1^t)_{t \in \Z}, E_1)$
and $(\Hil_2, (U_2^t)_{t \in \Z}, E_2)$ be $d$-dimensional smooth quantum walks.
Let $D_i^{(1)}$ be the self-adjoint operator $\int_{\R^r} x_i E_1(d \xx)$.
Let $D_i^{(2)}$ be the self-adjoint operator $\int_{\R^r} x_i E_2(d \xx)$.
Let $V \colon \Hil_1 \to \Hil_2$ be a smooth unitary operator which gives similarity between $U_1$ and $U_2$.
For $i = 1, \cdots, d$, 
denote by $\partial_i^{(1)}$ the $i$-th partial derivative with respect to $E_1$.
and 
denote by $\partial_i^{(2)}$ the $i$-th partial derivative with respect to $E_2$.
Then for $i = 1, \cdots, d$, 
the sequence 
\[\left\{ 
\frac{U_2^{-t} \partial_i^{(2)} (U_2^t)}{\ii t}
-
V \frac{U_1^{-t} \partial_i^{(1)} (U_1^t)}{\ii t} V^{-1}
\right\}_{t}
\]
converges to $0$ in the norm topology.
\end{proposition}

\begin{proof}
By Lemma \ref{lemma: partial derivative and commutator}, the operator $\frac{U_2^{-t} \partial_i^{(2)} (U_2^t)}{\ii t}$ is equal to the closure of 
\begin{eqnarray*}
\frac{V U_1^{-t} V^{-1} D_i^{(2)} V U_1^t V^{-1} - D_i^{(2)}}{t}.
\end{eqnarray*}
Since $V$ is smooth with respect to $E_1$ and $E_2$, 
by Lemma \ref{lemma: partial derivative and commutator} for $E_1 \oplus E_2$,
the operator $D_i^{(2)} V - V D_i^{(1)}$ is bounded.
This means that distance between $V^{-1} D_i^{(2)} V$ and $D_i^{(1)}$ is small.
If $t$ is large, the above operator is almost equal to
\begin{eqnarray*}
\frac{V U_1^{-t} D_i^{(1)} U_1^t V^{-1} - V D_i^{(1)} V^{-1}}{t}
=
\frac{V ( U_1^{-t} D_i^{(1)} U_1^t - D_i^{(1)}) V^{-1}}{t}.
\end{eqnarray*}
This is equal to the operator $V \frac{U_1^{-t} \partial_i^{(1)} (U_1^t)}{\ii t} V^{-1}$.
\end{proof}

\subsection{General theory for asymptotic behavior of the distribution $p_t$}
\label{subsection: asymptotic behavior of distributions}

For the quantum walk $\left( \Hil, \left(U^t \right)_{t \in \Z}, E \right)$, and for a unit vector $\xi \in \Hil$,
a sequence $\{p_t\}_{t \in \Z}$ of Borel probability measures on $\R^d$ is defined by
\begin{eqnarray}\label{equation: distribution of velocity}
p_t(\Omega) = \langle E(t \Omega) U^t \xi, U^t \xi\rangle, \quad \ {\rm for \ every \ Borel\ subset\ } \Omega \subset \R^d.
\end{eqnarray}
The unit vector $\xi \in \Hil$ is called an {\it initial vector}.
In many concrete examples of homogeneous quantum walks, the existence of the weak limit of $\{p_t\}$ has already been studied. See \cite{InuiKonishiKonno}, \cite{KonnoJMSJ}, \cite{WatabeKobayashiKatoriKonno}, for example.

\begin{example}
To see what the measure $p_t$ means, let us look at a quantum walk $U$ acting on $\ell_2(\Z^d) \otimes \C^n$.
Take the spectral measure $E$ as 
in Example \ref{example: standard spectral measure}.
Express $U^t \xi \in \ell_2(\Z^d) \otimes \C^n$ 
by $(\Psi_t(\xx))_{\xx \in \Z^d}$, where $\Psi_t(\xx)$ 
is a vector in $\C^n$.
By equation (\ref{equation: distribution of velocity}),
for every subset $\Omega$ of $\R^d$,
we have
\[p_t(\Omega) 
= \sum_{\xx \in t \Omega \cap \Z^d} \|\Psi_t(\xx)\|^2
= \sum_{\vv \in \Omega \cap \Z^d / t} \|\Psi_t(t \vv)\|^2.\]
Therefore,
the measure $p_t$ is a sum of scalar multiple of point masses
$\{\delta_\vv\}_{\vv \in \Z^d / t}$,
and the coefficient of $\delta_\vv$ is $\|\Psi_t(t \vv)\|^2$.
\qed
\end{example}

\begin{lemma}
\label{lemma: mean of a bounded Borel function wrt p_t}
Let $f$ be a bounded Borel function.
The mean of $f$ with respect to the measure $p_t$ is given by
\begin{eqnarray}
\label{equation: mean of a bounded Borel function wrt p_t}
\int_{\vv \in \R^d} f(\vv) p_t(d \vv) = 
\left\langle \int_{\xx \in \R^d} f \left( \frac{\xx}{t} \right) E(d \xx) U^t \xi, U^t \xi \right\rangle.
\end{eqnarray}
\end{lemma}

\begin{proof}
In the case that $f$ is a definition function of a Borel subset of $\R^d$,
the above equation holds, by the definition of $p_t$.
It follows that for every Borel step function $f$, the above equation holds.
For a general Borel function $f$, we have only to use a sequence $f_n$ of Borel step functions which are uniformly close to $f$.
\end{proof}

\begin{definition}
A vector $\xi$ in $\Hil$ is said to be smooth, if 
it is in the domain of $D_{i(m)} D_{i(m -1)} \cdots D_{i(1)}$ for every natural number $m$ and for every sequence $\{i(1)$, $\cdots$, $i(m)\}$ of $\{1, \cdots, d\}$.
\end{definition}

If the quantum walk $\left( \Hil, \left(U^t \right)_{t \in \Z}, E \right)$ is smooth and if the initial unit vector $\xi$ is smooth, then $U^t \xi$ is also smooth. The proof is given by Lemma \ref{lemma: partial derivative and commutator} and by the Leibniz rule:
$D_i(U^t \xi) = -\ii \partial_i(U^t) \xi + U^t (D_i \xi)$.

\begin{lemma}
\label{lemma: mean of a polynomial function wrt p_t}
Suppose that the quantum walk $\left( \Hil, \left(U^t \right)_{t \in \Z}, E \right)$ is smooth and that the vector $\xi$ is smooth.
Let $t$ be an arbitrary natural number.
Then every polynomial function $g$ on $\R^d$ is integrable with respect to $p_t$,
and the integral is given by
\[\int_{\vv \in \R^d} g(\vv) p_t(d \vv) = 
\left\langle \int_{\xx \in \R^d} g \left( \frac{\xx}{t} \right) E(d \xx) U^t \xi, U^t \xi \right\rangle.
\]
\end{lemma}

\begin{proof}
We may assume that $g(x_1, x_2 \cdots, x_d)$ is of the form $x_{i(m)} x_{i(m -1)} \cdots x_{i(1)}$.
In this case, we have
\[
\int_{\xx \in \R^d} g \left( \frac{\xx}{t} \right) E(d \xx)
=
\frac{D_{i(m)}}{t} \cdots \frac{D_{i(1)}}{t}.
\]
Since $U$ and $\xi$ is smooth with respect to $E$,
$U^t \xi$ is in the domain of the above self-adjoint operator.
Therefore the right hand side is well-defined.
There exists a countable sum $f(\vv)$ of scalar multiples of definition functions which is Borel and uniformly close to $g(\vv)$.
Then the difference between
\[
\int_{\xx \in \R^d} g \left( \frac{\xx}{t} \right) E(d \xx)
\sim
\int_{\xx \in \R^d} f \left( \frac{\xx}{t} \right) E(d \xx),
\]
is a bounded operator with small operator norm.
Therefore $U^t \xi$ is in the domain of $\int_{\xx \in \R^d} f \left( \frac{\xx}{t} \right) E(d \xx)$.
For such $f$, it is easy to show that
\[\int_{\vv \in \R^d} f(\vv) p_t(d \vv) = 
\left\langle \int_{\xx \in \R^d} f \left( \frac{\xx}{t} \right) E(d \xx) U^t \xi, U^t \xi \right\rangle.
\]
Since $f$ is uniformly close to $g$, the integral
$\int_{\vv \in \R^d} f(\vv) p_t(d \vv)$ is close to $\int_{\vv \in \R^d}$ $g(\vv)$ $p_t(d \vv)$.
\end{proof}

\begin{theorem}
\label{theorem: asymptotic concentration}
Suppose that the quantum walk $\left( \Hil, \left(U^t \right)_{t \in \Z}, E \right)$ is smooth and that the initial unit vector $\xi$ is smooth.
There exists a bounded subset $K$ of $\R^d$ such that $\lim_{t \to \infty} p_t(K) = 1$.
\end{theorem}

\begin{proof}
Let $i$ be an arbitrary element of $\{1, \cdots, d\}$.
Let $m$ be an arbitrary natural number.
Note that the equation
\begin{eqnarray*}
\int_{(v_1, \cdots, v_d) \in \R^d} v_i^m p_t(d v_1 \cdots d v_i \cdots v_d)
=
\left\langle \frac{D_i^m}{t^m} U^t \xi, U^t \xi \right\rangle
\end{eqnarray*}
holds, by Lemma \ref{lemma: mean of a polynomial function wrt p_t}.

The triplet $(\Hil, \left( U^t \right)_{t \in \Z}, D_i)$ is a one-dimensional quantum walk defined in \cite[Definition 2.1]{Sako}. Therefore we can employ the theory of one-dimensional quantum walks developed in \cite{Sako}.
As proved in \cite[Definition 2.27]{Sako}, for every $m \in \N$,
\begin{eqnarray*}
\limsup_t \left| \left\langle \frac{D_i^m}{t^m} U^t \xi, U^t \xi \right\rangle\right|
 \le \|[D_i, U]\|^m.
\end{eqnarray*}
Take and fix a positive number $L_i$ larger than $\|[D_i, U]\|$.
The inequality 
\begin{eqnarray*}
\limsup_t \left| \int_{(v_1, \cdots, v_d) \in \R^d} v_i^m p_t(d v_1 \cdots d v_i \cdots v_d) \right|
 \le \|[D_i, U]\|^m.
\end{eqnarray*}
implies that
$\lim_t p_t(\R^{i - 1} \times [- L_i, L_i] \times \R^{d - i}) = 1$.
We conclude that $\lim_t p_t([- L_1, L_1] \times \cdots \times [- L_d, L_d]) = 1$.
\end{proof}

\begin{corollary}
\label{corollary: limit distribution}
Suppose that the quantum walk $\left( \Hil, \left(U^t \right)_{t \in \Z}, E \right)$ is smooth and that the initial unit vector $\xi$ is smooth.
Then the following conditions are equivalent:
\begin{enumerate}
\item
(Convergence in Law).
There exists a Borel probability measure $p_\infty$ on $\R_d$ such that for every polynomial function $g$ on $\R^d$
\[\lim_t  \int_{\vv \in \R^d} g(\vv) p_t(d \vv) = 
\int_{\vv \in \R^d} g(\vv) p_\infty(d \vv).\]
\item
(Weak convergence).
There exists a Borel probability measure $p_\infty$ on $\R_d$ such that for every bounded continuous function $f$ on $\R^d$
\[\lim_t  \int_{\vv \in \R^d} f(\vv) p_t(d \vv) = 
\int_{\vv \in \R^d} f(\vv) p_\infty(d \vv).\]
\end{enumerate}
If the above conditions hold, then these limit distributions coincide and their support is compact.
\end{corollary}

\begin{proof}
Let $K$ be a compact subset of $\R^d$ 
in Theorem \ref{theorem: asymptotic concentration}.
Let $\epsilon$ be an arbitrary positive real number.
For every polynomial function $g$ on $\R^d$, 
there exists a bounded continuous function $f$ on $\R^d$ such that $|f(\xx) - g(\xx)| < \epsilon$ for arbitrary $\xx \in K$.
For every bounded continuous function $f$ on $\R^d$, 
there exists a polynomial function $g$ on $\R^d$ such that $|g(\xx) - f(\xx)| < \epsilon$ for arbitrary $\xx \in K$ by the theorem of Stone-Weierstrass.
\end{proof}

\begin{proposition}
\label{proposition: similar asymptotic behavior}
Let
$(\Hil_1, (U_1^t)_{t \in \Z}, E_1)$, $(\Hil_2, (U_2^t)_{t \in \Z}, E_2)$ be
two $d$-dimensional smooth quantum walks.
Suppose that there exists a smooth unitary operator $V \colon \Hil_1 \to \Hil_2$ which intertwines $U_1$ and $U_2$.
Let $\xi$ be a unit vector in $\Hil_1$ which is smooth with respect to $E_1$.
Let $p_t^{(1)}$ be the sequence of probability measures on $\R^d$ given by $(\Hil_1, (U_1^t)_{t \in \Z}, E_1)$ and $\xi$.
Let $p_t^{(2)}$ be the sequence of probability measures on $\R^d$ given by $(\Hil_2, (U_2^t)_{t \in \Z}, E_2)$ and $V \xi$.
If $p_t^{(1)}$ converges in law, then $p_t^{(2)}$ also converges in law.
Furthermore, these limit distributions coincide.
\end{proposition}

\begin{proof}
The proof is substantially the same as that of \cite[Theorem 2.31]{Sako}.
\end{proof}

\section{Convergence theorems on homogeneous QWs}

\subsection{Definition of homogeneous QWs}
\label{subsection: homogeneous QWs}

\begin{definition}
A quadruple $(\Hil, \left( U^t \right)_{t \in \Z}, E, \rho)$
is called a $d$-dimensional {\rm homogeneous} quantum walk, if the following conditions holds:
\begin{enumerate}
\item
$\left( \Hil, \left(U^t \right)_{t \in \Z}, E \right)$ is a $d$-dimensional quantum walk.
\item
$\rho = (\rho(\xx))_{\xx \in \Z^d}$ is a unitary representation of the additive group $\Z^d$ on $\Hil$.
\item
For every Borel subset $\Omega \subset \R^d$, $\rho(\xx)^{-1} E(\Omega) \rho(\xx) = E(\Omega + \xx)$.
\item
For every $\xx \in \Z^d$, $U \rho(\xx) = \rho(\xx) U$.
\item
The rank of $E \left( [0, 1)^d \right)$ is finite.
\end{enumerate}
The rank of $E \left( [0, 1)^d \right)$ is called the {\rm degree of freedom}.
\end{definition}

We note that the image of $E \left( [0, 1)^d \right)$ is the fundamental domain of the action by $(\rho(\xx))_{\xx \in \Z^d}$. More precisely, $\{\rho(\xx) (\mathrm{image}(E([0, 1)^d))) \}_{\xx \in \Z^d}$ is mutually orthogonal and generates $\Hil$.
We may consider a quadruple satisfying conditions $(1)$, $(2)$, $(3)$, $(4)$ which does not satisfy condition $(5)$.
In such a case, we call the quadruple {\it a homogeneous quantum walk with infinite degree of freedom}.

Let $(\Hil, \left( U^t \right)_{t \in \Z}, E, \rho)$ be a homogeneous quantum walk with finite degree of freedom.
We denote by $n$ the degree of freedom.
There exists a unitary operator
\[V \colon \Hil \to \ell_2(\Z^d) \otimes \C^n\]
satisfying that $V$ maps $\mathrm{image}(E([0, 1)^n))$ onto $\delta_{\mathbf{0}} \otimes \C^n$ and that $V$ is compatible with $\rho$ and with the right regular representation $\widetilde{\rho}$ on $\ell_2(\Z^d)$.
Let $\widetilde{E}$ be the projection-valued measure on $\R^d$ defined by 
\[\widetilde{E}(\Omega) = (\textrm{the \ orthogonal \ projection\ } \ell_2(\Z^d) \otimes \C^n \to \ell_2(\Z^d \cap \Omega) \otimes \C^n ).\]
It is easy to show that $V$ is smooth with respect to $E$ and $\widetilde{E}$.
In fact, $V$ is analytic with respect to $E$ and $\widetilde{E}$.
Thus we obtain a new homogeneous quantum walk
\[\left( \ell_2(\Z^d) \otimes \C^n, (V U^t V^{-1})_{t \in \Z}, \widetilde{E}, \widetilde{\rho} \right),\]
which is similar to the original walk
in the sense of Definition \ref{definition: similarity}.
If the original walk is smooth, then new one is also smooth.
If the original walk is analytic, then new one is also analytic.
By Proposition \ref{proposition: similar asymptotic behavior},
if the latter walk has limit distribution (as in Theorem \ref{theorem: limit distribution}), the original walk has the same limit distribution.

For the rest of this paper, we study $d$-dimensional {\it analytic homogeneous} quantum walks.
Without loss of generality, we may concentrate on the homogeneous walks of the form $(\ell_2(\Z^d) \otimes \C^n, \left( U^t \right)_{t \in \Z}, \widetilde{E}, \widetilde{\rho})$.

\subsection{The proof of the convergence theorem for homogeneous QWs}

We often use the inverse Fourier transform
$\F^{-1} \colon \ell_2(\Z^d) \otimes \C^n \to L^2 \left( \T_{2 \pi}^d \right) \otimes \C^n$.
We express the Pontryagin dual $\T_{2 \pi}^d$ of $\Z^d$ by $\{(k_1, \cdots, k_d) \ |\ k_1, \cdots, k_d \in \R / (2 \pi \Z)\}$.
Via the inverse Fourier transform $\F^{-1}$,
\begin{itemize}
\item
the analytic unitary operator $U$ corresponds to a $(d \times d)-$matrix $\widehat{U}$ whose entries are analytic functions on $\T_{2 \pi}^d$,
\item
and the diagonal operator
\[D_i = \int_{\xx \in \R^d} x_i \widetilde{E}(d x_1 \cdots d x_i \cdots d x_d) \colon \delta_{\xx} \otimes \delta_y \mapsto x_i \delta_{\xx} \otimes \delta_y\]
corresponds to the partial differential operator $- \ii \frac{\partial}{\partial k_i}$.
\end{itemize}

Our main result (Theorem \ref{theorem: limit distribution})
relies on the following proposition.

\begin{proposition}\label{proposition: convergence of LD for homogeneous QW}
Let $\left( U^t \right)_{t \in \Z}$
be a homogeneous analytic quantum walk acting on $\ell_2(\Z^d) \otimes \C^n$.
Let $w_1, \cdots, w_d$ be real numbers.
Denote by $D$ the operator $\sum_i w_i D_i$.
\begin{enumerate}
\item
As $t$ tends to infinity,
the average $c_t / t = (U^{-t} D U^t - D) / t$ of the logarithmic derivative of $\left( U^t \right)_{t \in \Z}$ converges to some self-adjoint operator $H$ {\rm in the strong operator topology}.
\item
Then as $t$ tends to infinity, the unitary operators defined by the commutators
\[U^{-t} \exp \left( \ii \frac{D}{t} \right) U^t \exp \left( - \ii \frac{D}{t} \right) \]
converge to the unitary operator $\exp(\ii H)$ {\rm in the strong operator topology}.
\end{enumerate}
For the case that the walk is one-dimensional $(d = 1)$,
the above two sequences converge {\rm in norm}, and the limits are analytic operators.
\end{proposition}

To employ the analytic perturbation theory by Tosio Kato,
we suppose analyticity.
We denote by $\mathrm{diag}(\alpha_i)_i$ the diagonal matrix
whose diagonal entries are $\alpha_1$, $\cdots$, $\alpha_n$.

\begin{proof}
We make use of the Fourier transform $\widehat{U}(\kk), \kk \in \T_{2 \pi}^d$ of the walk.
The operator $D$ corresponds 
to the operator $\widehat{D} = - \ii \sum_i w_i \frac{\partial}{\partial k_i}$.
Denote by $\ww$ the real vector $(w_i)$.
By Lemma \ref{lemma: partial derivative and commutator}, 
the average of the logarithmic derivative $c_t / t = U^{-t} [D, U^t] / t$ corresponds to
\begin{eqnarray}\label{equation: ALD and directional derivative}
\widehat{c_t} / t=
\frac{1}{t} \widehat{U}^{-t} \left[\widehat{D}, \widehat{U}^t \right]
=
\frac{1}{\ii t} \widehat{U}^{-t} \sum_i w_i \frac{\partial  \widehat{U}^t}{\partial k_i}.
\end{eqnarray}

We make use of
the directional derivative $\frac{d}{d s} = \sum_i w_i \frac{\partial}{\partial k_i}$ along the curve of the form
$\R \ni s \mapsto \kk_0 + s \ww \in \T_{2 \pi}^d$.
Choose an arbitrary base point $\kk_0 \in \T_{2 \pi}^d$.
Let $I$ be a closed interval in $\R$ whose interior part includes $0$.
For a while, we study the behavior of $\widehat{U}$ on the segment $\kk_0 + I \ww \subset \T_{2 \pi}^d$.
Making the interval $I$ shorter, we may assume that the map
\[I \ni s \mapsto \kk_0 + s \ww \in \T_{2 \pi}^d.\]
is injective.
Denote by $\widehat{u}(s)$ the unitary $\widehat{U}(\kk_0 + s \ww)$.
Note that the map $I \ni s \mapsto \widehat{u}(s)$ can be extended to a holomorphic map defined on a complex domain including $I$.
By equation (\ref{equation: ALD and directional derivative}),
The average of the logarithmic derivative $\widehat{c_t} / t$ is equal to
\begin{eqnarray}\label{equation: ALD and directional derivative 2}
\frac{\widehat{c_t}(\kk_0 + s \ww)}{t} =
\frac{1}{\ii t} \widehat{u}^{-t}(s) \frac{d  }{d s} (\widehat{u}^t(s)).
\end{eqnarray}

Let us make use of the analytic perturbation theory by T. Kato.
The unitary $\widehat{u}(s)$ can be decomposed as follows
\begin{eqnarray}
\label{equation: diagonal decomposition}
\widehat{u}(s)
=
\widehat{v}(s)
\cdot
\mathrm{diag}(\lambda_1(s), \cdots, \lambda_n(s))
\cdot
\widehat{v}(s)^{-1},
\end{eqnarray}
by \cite[Theorem II.1.8 and II.1.10]{TKato}.
See also \cite[Section II.4.6]{TKato}.
Here $\widehat{v}$ is an analytic map to invertible matrices,
and the eigenvalue functions $\lambda_1$, $\cdots$, $\lambda_n$ are holomorphic functions defined on a complex domain including $I$. 

A direct calculation yields that
\begin{eqnarray*}
\frac{d}{d s} (\widehat{u}(s)^t)
&=&
\frac{d \widehat{v}}{d s}(s) \cdot 
\mathrm{diag} \left( \lambda_i(s)^t \right) \cdot 
\widehat{v}(s)^{-1} \\
&& \qquad +
\widehat{v}(s) \cdot \mathrm{diag} \left( \dfrac{d}{d s} (\lambda_i(s)^t) \right) \cdot \widehat{v}(s)^{-1} \\
&& \qquad \qquad -
\widehat{v}(s) \cdot \mathrm{diag}(\lambda_i(s)^t) 
\cdot \widehat{v}(s)^{-1} \cdot \frac{d \widehat{v}}{d s}(s) \cdot \widehat{v}(s)^{-1}.
\end{eqnarray*}
The second term is equal to 
\[
t \cdot \widehat{v}(s) 
\cdot \mathrm{diag} 
\left( \lambda_i(s)^{t - 1} \cdot \dfrac{d \lambda_i}{d s}(s) \right) 
\cdot \widehat{v}(s)^{-1}.
\]
As $t$ tends to infinity,
the norm increases linearly.
The norms of the first and the third terms are bounded.
It follows that for the calculation of the average with respect to time,
it suffices to see the second term.
By equation (\ref{equation: ALD and directional derivative 2}),
we have
\begin{eqnarray}
&&
\lim_{t \to \infty} 
\frac{\widehat{c_t}(\kk_0 + s \ww)}{t}\\
&=&
\widehat{v}(s)
\cdot
\mathrm{diag}(\lambda_i(s)^t)
\cdot
\widehat{v}(s)^{-1}
\cdot 
\widehat{v}(s) 
\cdot 
\mathrm{diag} 
\left( \lambda_i(s)^{t - 1} \cdot \dfrac{d \lambda_i}{d s}(s) \right) 
\cdot 
\widehat{v}(s)^{-1}\\
&=&\label{equation: limit of LD on segment}
\widehat{v}(s)
\cdot
\mathrm{diag} 
\left( \lambda_i(s)^{- 1} \cdot \dfrac{d \lambda_i}{d s}(s) \right) 
\cdot 
\widehat{v}(s)^{-1}.
\end{eqnarray}
The convergence is uniform on the closed interval $I$.
Because $|\lambda_i(s)|^2 = 1$, $\lambda_i(s)^{- 1} \cdot \dfrac{d \lambda_i}{d s}(s)$ is in $\ii \R$, as $t$ tends to infinity, the average of logarithmic derivatives $\widehat{c_t}(\kk_0) / t$ converges to a skew self-adjoint matrix $\ii \widehat{H}(\kk_0)$, at each point $\kk_0 \in \T_{2 \pi}^d$.
The operator norm is uniformly bounded by $\left\| \sum_i w_i \frac{\partial \widehat{U}}{\partial k_i} \right\|$ due to Lemma \ref{lemma: bounded by c1}.
This upper bound is independent of $\kk_0$.
Therefore, the convergence of $\widehat{c_t}(\kk_0) / t$ 
at each point $\kk_0$ yields that in the strong operator topology.
Applying the Fourier transform, 
we obtain the first assertion.

Consider the case that $d$ is one.
Then the convergence on each closed interval in $\T_{2 \pi}$ is uniform,
and the limit is given by a matrix-valued analytic function
as in equation (\ref{equation: limit of LD on segment}).
Since the one-dimensional torus  $\T_{2 \pi}$ is a union of two segments,
the average of logarithmic derivative converges in norm and the limit is analytic.

Let us consider general $d$ again.
The inverse Fourier transform of $U^{-t} \exp \left( \ii \frac{D}{t} \right) U^t \exp \left( - \ii \frac{D}{t} \right)$ is
\[\widehat{U}^{-t} \exp \left( t^{-1} \sum_i w_i \frac{\partial}{\partial k_i} \right) 
\widehat{U}^t \exp \left( - t^{-1} \sum_i w_i \frac{\partial}{\partial k_i} \right).\]
For every $\C^n$-valued analytic functions $\widehat{\xi}(\kk)$ defined on $\T_{2 \pi}^d$,
we calculate
\[\widehat{U}^{-t} \exp \left( t^{-1} \sum_i w_i \frac{\partial}{\partial k_i} \right) 
\widehat{U}^t \exp \left( - t^{-1} \sum_i w_i \frac{\partial}{\partial k_i} \right) 
\widehat{\xi}\]
as follows.
The $\C^n$-valued function $\exp \left( - t^{-1} \sum_i w_i \frac{\partial}{\partial k_i} \right) \widehat{\xi}$ is given by
\begin{eqnarray*}
\left[ \exp \left( - t^{-1} \sum_i w_i \frac{\partial}{\partial k_i} \right) \widehat{\xi} \right]
(\kk)
=
\sum_{m = 0}^\infty \frac{1}{m !} 
t^{-m} \left[ \left(\sum_i w_i \frac{\partial}{\partial k_i} \right)^m \widehat{\xi} \right] (\kk).
\end{eqnarray*}
This is the Taylor expansion of $\xi(\kk - t^{-1} \ww)$.
Therefore, the unitary $\exp \left( - t^{-1} \sum_i \frac{\partial}{\partial k_i} \right)$ acts on analytic vectors by the translation of $t^{-1} \ww$.
Since analytic vectors are dense in the Hilbert space,
the unitary $\exp \left( - t^{-1} \sum_i \frac{\partial}{\partial k_i} \right)$ acts on every vector in $L^2(\T_{2 \pi}^d) \otimes \C^n$
by the translation of $t^{-1} \ww$.
By the same reason, the unitary $\exp \left( t^{-1} \sum_i \frac{\partial}{\partial k_i} \right)$ means the translation by $- t^{-1} \ww$.
It follows that for every $\widehat{\xi} \in L^2 \left( \T_{2 \pi}^d \right) \otimes \C^n$, the following equation holds:
\begin{eqnarray*}
&&
\left[ \widehat{U}^{-t} \exp \left( t^{-1} \sum_i \frac{\partial}{\partial k_i} \right) 
\widehat{U}^t \exp \left( - t^{-1} \sum_i \frac{\partial}{\partial k_i} \right) \widehat{\xi} \right] (\kk)\\
&=&
\widehat{U}(\kk)^{-t} \cdot
\widehat{U}(\kk + t^{-1} \ww)^t \cdot \widehat{\xi}(\kk).
\end{eqnarray*}
Since the vector 
$\widehat{\xi}$
in $L^2(\T_{2 \pi}^d) \otimes \C^n$ is arbitrary,
we obtain the following equation between two operators:
\begin{eqnarray}
&&
\left[ \widehat{U}^{-t} 
\exp \left( t^{-1} \sum_i \frac{\partial}{\partial k_i} \right) 
\widehat{U}^t 
\exp \left( - t^{-1} \sum_i \frac{\partial}{\partial k_i} \right) \right] (\kk)\\
&=&
\widehat{U}(\kk)^{-t} \cdot
\widehat{U}(\kk + t^{-1} \ww)^t.
\label{eqnarray: commutator and translation}
\end{eqnarray}

We make use of the restriction $\widehat{u}(s) = \widehat{U}(\kk + s \ww)$ on the segment $\kk_0 + I \ww \subset \T_{2 \pi}^d$
and its diagonal decomposition $\widehat{v}(s) \cdot \mathrm{diag}(\lambda_i(s))_i \cdot \widehat{v}(s)^{-1}$
as in equation (\ref{equation: diagonal decomposition}).
Assume that $s$ is in $I$.
Since $I$ is simply connected, there exist real-valued analytic functions $l_1$, $l_2$, $\cdots$, $l_n$ such that
$\exp( \ii l_i(s)) = \lambda_i(s)$.

The formula (\ref{eqnarray: commutator and translation}) is equal to 
$\widehat{u}(0)^{-t}
\widehat{u}(t^{-1})^t$.
If $t$ is large, then the matrix $\widehat{v}(t^{-1})$ is close to $\widehat{v}(0)$ in norm.
The unitary matrix
$\widehat{u}(0)^{-t}
\widehat{u}(t^{-1})^t$
is uniformly close to
\begin{eqnarray*}
&&
\widehat{v}(0) \cdot 
\mathrm{diag}(\lambda_i(0)^{-t})_i 
\cdot \widehat{v}(0)^{-1} \cdot
\widehat{v}(t^{-1}) \cdot 
\mathrm{diag}(\lambda_i(t^{-1})^t)_i \cdot 
\widehat{v}(t^{-1})^{-1}\\
&\sim&
\widehat{v}(0) \cdot \mathrm{diag}(\lambda_i(0)^{-t})_i \cdot
\mathrm{diag}(\lambda_i(t^{-1})^t)_i \cdot 
\widehat{v}(t^{-1})^{-1}\\
&=&
\widehat{v}(0)
\mathrm{diag}\left(\exp \left( \ii \frac{l_i(t^{-1}) - l_i(0)}{t^{-1}} \right) \right)_i
\widehat{v}(0)^{-1}.
\end{eqnarray*}
Note that the logarithms $l_1(s)$, $\cdots$, $l_n(s)$ are differentiable,
since $\lambda_i(s)$ are analytic.

Therefore we have
\begin{eqnarray*}
\lim_{t \to \infty } \widehat{u}(0)^{-t}
\widehat{u}(t^{-1})^t 
&=& 
\widehat{v}(0)
\mathrm{diag}\left(\exp\left( \ii \dfrac{d l_i}{d s}(0) \right) \right)_i
\widehat{v}(0)^{-1}
\\
&=& 
\widehat{v}(0)
\mathrm{diag}\left(\exp\left( \lambda_i(0)^{-1} \dfrac{d \lambda_i}{d s}(0) \right) \right)_i
\widehat{v}(0)^{-1}\\
&=& 
\exp\left( \widehat{v}(0)
\mathrm{diag}\left(\lambda_i(0)^{-1} \dfrac{d \lambda_i}{d s}(0) \right)_i
\widehat{v}(0)^{-1}
\right).
\end{eqnarray*}
By equation (\ref{equation: limit of LD on segment}) and by the definition of the skew self-adjoint matrix $\ii \widehat{H}(\kk_0)$, we have the following convergence at each point $\kk_0$:
\begin{eqnarray*}
\lim_{t \to \infty}
\widehat{U}(\kk_0)^{-t} \cdot
\widehat{U}(\kk_0 + t^{-1} \ww)^t
&=&
\exp \left( \ii \widehat{H}(\kk_0) \right).
\end{eqnarray*}
Because the sequence consists of unitary operators, the convergence at each point implies that in the strong operator topology.
Applying the Fourier transform, we obtain the second assertion.

In the case that $d$ is one, the convergence is that in the operator norm topology, since the convergence is uniform on each segment included in the torus.
\end{proof}

In case that $d$ is one, the following theorem is proved in \cite{SaigoSako}.

\begin{theorem}
\label{theorem: limit distribution}
For any $d$-dimensional homogeneous analytic quantum walk with finite degree of freedom,
and for any initial unit vector, 
the weak limit of probability measures $\{p_t\}$ defined in Subsection
\ref{subsection: asymptotic behavior of distributions} exists.
\end{theorem}

The limit distribution is described as follows.
Let $(U^t)_{t \in \Z}$ be a homogeneous quantum walk acting on $\ell_2(\Z^d) \otimes \C^n$.
Let $\xi \in \ell_2(\Z^d) \otimes \C^n$ be an initial unit vector.
By Proposition \ref{proposition: convergence of LD for homogeneous QW}, 
the limit of the average of logarithmic derivatives
\[H_i = \lim_{t \to \infty} \frac{1}{\ii t} U^{-t} \partial_i(U^t) \]
exists and is a bounded self-adjoint operator.
Since the self-adjoint operators $(D_1$, $\cdots$, $D_d)$ mutually commute, the self-adjoint operators $(H_1$, $\cdots$, $H_d)$ also mutually commute.
There exists a unique projection-valued Borel probability measure $\mathcal{E}$ on a compact subset of $\R^d$
satisfying that for every integers $m(1), \cdots, m(d) \in \Z_{\le 0}$, 
\[\prod_i H_i^{m(i)} 
= \int_{\vv \in \R^d} \prod_i v_i^{m(i)} \mathcal{E}(d \vv).\]
The following argument shows that the limit distribution of the walk is equal to 
$p_\infty( \ \cdot \ ) = \langle \mathcal{E}( \ \cdot \ ) \xi, \xi \rangle$.

\begin{proof}
We first assume that $\xi$ is smooth.
Via the inverse Fourier transform, the walk $U$ corresponds to an $(n \times n)$-matrix $\widehat{U}$ whose entries are analytic functions on $\T_{2 \pi}^d$.
The initial unit vector $\xi$ corresponds to a smooth element $\widehat{\xi} \in L^2 \left( \T_{2 \pi}^d \right) \otimes \C^n$.
For every $i = 1, \cdots, d$, the self-adjoint operator $\widehat D_i$ is given by $- \ii \frac{\partial}{ \partial k_i}$.
Define $\widehat D$ by $- \ii \sum_i w_i \frac{\partial}{\partial k_i}$.
By equation (\ref{equation: mean of a bounded Borel function wrt p_t})
in Lemma \ref{lemma: mean of a bounded Borel function wrt p_t},
The mean of the function $\exp(\ii(\cdot, \ww)_{\R^d})$ on $\R^d$ with respect to $p_t$ is equal to the following inner product:
\begin{eqnarray*}
& &
\int_{\vv \in \R^d} \exp(\ii(\vv, \ww)_{\R^d})p_t(d \vv)\\
&=&
\left\langle \exp \left( \ii t^{-1} \widehat{D} \right) \widehat{U}^t \widehat{\xi}, \widehat{U}^t \widehat{\xi} \right\rangle\\
&=&
\left\langle \widehat{U}^{-t} \exp \left( \ii t^{-1} \widehat{D} \right) \widehat{U}^t  \exp \left( -\ii t^{-1} \widehat{D} \right) \exp \left( \ii t^{-1} \widehat{D} \right)\widehat{\xi}, \widehat{\xi} \right\rangle.
\end{eqnarray*}
Let $\widehat{H}$ be the limit of the averages of logarithmic derivatives of $\widehat{U}^t$ with respect to the differential operator $- \ii \sum w_i \frac{\partial}{\partial k_i}$.
As in the proof of Proposition 
\ref{proposition: convergence of LD for homogeneous QW},
define self-adjoint operators $\widehat{H_i}$ by
\[
- \ii \lim_{t \to \infty} 
\widehat{U}(\kk)^{-t} \frac{\partial \widehat{U}^t}{\partial k_i}(\kk),
\]
and $\widehat{H}$ by $\sum_{i=1}^d w_i \widehat{H_i}$.
By Proposition \ref{proposition: convergence of LD for homogeneous QW},
the operator 
\[\widehat{U}^{-t} \exp \left( \ii t^{-1} \widehat{D} \right) \widehat{U}^t  \exp \left( -\ii t^{-1} \widehat{D} \right)\]
converges to $\exp \left( \ii \widehat{H} \right)$ in the strong operator topology.
The vector $\exp \left( \ii t^{-1} \widehat{D} \right)\widehat{\xi}$ is given by
$\left[\exp \left( \ii t^{-1} \widehat{D} \right)\widehat{\xi} \right](\kk) = \widehat{\xi}(\kk + t^{-1} \ww)$.
It is uniformly close to $\widehat{\xi}(\kk)$.
As $t$ tends to infinity, the above integral converges to
\begin{eqnarray*}
\left\langle \exp \left( \ii \widehat{H} \right) \widehat{\xi}, 
\widehat{\xi} \right\rangle_{L^2 \left( \T_{2 \pi}^d \right) \otimes \C^d}
=
\left\langle \prod_i \exp(\ii w_i H_i) \xi, \xi \right\rangle.
\end{eqnarray*}
We obtain
\begin{eqnarray*}
\lim_{t \to \infty} \int_{\vv \in \R^d} \exp(\ii(\vv, \ww)_{\R^d})p_t(d \vv) 
&=& \left\langle \prod_i \exp(\ii w_i H_i) \xi, \xi \right\rangle\\
&=& \left\langle \int_{\vv \in \R^d} \exp(\ii (\vv, \ww)_{\R^d}) \mathcal{E}(d \vv) \xi, \xi \right\rangle.
\end{eqnarray*}

We obtain that for for every linear combination $g$ of 
$\{ \exp(\ii(\cdot, \ww)_{\R^d}) \ |\ \ww \in \R^d \}$,
\[\lim_{t \to \infty} \int_{\vv \in \R^d}  g(\vv) p_t(d \vv) 
= 
\left\langle \int_{\vv \in \R^d} g(\vv) \mathcal{E}(d \vv) \xi, \xi \right\rangle.\]
By Theorem \ref{theorem: asymptotic concentration},
there exists a compact subset $K$ of $\R^d$ such that
\[\lim_{t \to \infty} p_t(K) = 1.\]
The linear span $\mathrm{span}\{ \exp(\ii(\cdot, \ww)_{\R^d}) \ |\ \ww \in \R^d \}$ is the space of trigonometric functions.
By the theorem of Stone--Weierstrass,
the linear span is dense in $C(K)$ with respect to the supremum norm.
It follows that for every bounded continuous function $f$ on $\R^d$,
\[\lim_{t \to \infty} \int_{\vv \in \R^d} f(\vv) p_t(d \vv) 
= 
\left\langle \int_{\vv \in \R^d} f(\vv) \mathcal{E}(d \vv) \xi, \xi \right\rangle.\]
In the special case that $\xi$ is smooth, we finish the proof.

Let $\widetilde{\xi}$ be an initial unit vector in $\ell_2(\Z^d) \otimes \C^n = \ell_2(\Z^d \to \C^n)$ which is not necessarily smooth.
Let $\widetilde{p_t}$ be the sequence of probability measures defined 
by $U^t$ and $\widetilde{\xi}$.
Let $\epsilon$ be an arbitrary positive number.
Choose a smooth unit vector $\xi \in \ell_2(\Z^d \to \C^n)$ satisfying that $\left\|\xi - \widetilde{\xi} \right\| < \epsilon$.
For every element $\eta$ of $\ell_2(\Z^d \to \C^n)$,
define an $\ell_1$ function $|\eta|^2$ on $\Z^d$ by
\[
\left| \eta \right|^2(\xx) 
= \left\| \eta(\xx) \right\|_{\C^n}^2, 
\]
By the inequality $\left\| U^t \xi - U^t \widetilde{\xi} \right\| < \epsilon$,
we have
$\left\| |U^t \xi|^2 - \left|U^t \widetilde{\xi} \right|^2 \right\|_{\ell_1} < 2 \epsilon$.
By equation (\ref{equation: mean of a bounded Borel function wrt p_t}),
for every bounded Borel function $f$ on $\R^d$, we have
\begin{eqnarray*}
\left|
\int_{\vv \in \R^d} f(\vv) p_t(d \vv)
-
\int_{\vv \in \R^d} f(\vv) \widetilde{p_t} (d \vv)
\right|
&\le&
\sup_{\vv \in \R^d} |f(\vv)| \cdot \left\| |U^t \xi|^2 
- \left| U^t \widetilde{\xi} \right|^2 \right\|_{\ell_1} \\
&\le& 
2 \epsilon \sup_{\vv \in \R^d} |f(\vv)|.
\end{eqnarray*}
We also obtain
\begin{eqnarray*}
&&
\left|
\left\langle \int_{\vv \in \R^d} f(\vv) \mathcal{E}(d \vv) \xi, \xi \right\rangle
-
\left\langle \int_{\vv \in \R^d} f(\vv) \mathcal{E}(d \vv) \widetilde{\xi}, 
\widetilde{\xi} \right\rangle
\right|\\
&\le&
2
\left\|
\int_{\vv \in \R^d} f(\vv) \mathcal{E}(d \vv)
\right\|
\left\|\xi - \widetilde{\xi} \right\|
\le 
2 \epsilon \sup_{\vv \in \R^d} |f(\vv)|.
\end{eqnarray*}
It follows that for every bounded continuous function $f$ on $\R^d$,
\[\lim_{t \to \infty} \int_{\vv \in \R^d} f(\vv) \widetilde{p_t}(d \vv) 
= 
\left\langle \int_{\vv \in \R^d} f(\vv) \mathcal{E}(d \vv) \widetilde{\xi}, \widetilde{\xi} \right\rangle
= 
\int_{\vv \in \R^d} f(\vv) \left\langle \mathcal{E}(d \vv) \widetilde{\xi}, \widetilde{\xi} \right\rangle.\]
It follows that the sequence of probability measures $\left\{ \widetilde{p_t} \right\}$ weakly converges.
\end{proof}

\subsection{A quantum walk with an initial unit vector whose support is not localized.}
The convergence theorem 
(Theorem \ref{theorem: limit distribution}) holds true for an arbitrary initial unit vector.
Let us consider an example, in which the support of the initial unit vector $\xi$ is the whole space $\Z$.
We define $U$ by the $3$-state Grover walk
\[
U =
\left(
\begin{array}{ccc}
S & 0 & 0\\
0 & 1 & 0\\
0 & 0& S^{-1}
\end{array}
\right)
\cdot \frac{1}{3}
\left(
\begin{array}{rrr}
-1 & 2 & 2\\
2 & -1 & 2\\
2 & 2 & -1
\end{array}
\right)
\]
acting on $\ell_2(\Z) \otimes \C^3$.
We set the initial unit vector $\xi$ by the infinite sum
\[\xi = 
\frac{1}{\sqrt 3}
\sum_{x \in \Z} \delta_x 
\otimes 
\left(
\begin{array}{c}
0\\
2^{- |x| / 2} \\
0
\end{array}
\right).
\]
When we regard $\xi$ as a map from $\Z$ to $\C^3$,
the support of $\xi$ is $\Z$.
However, according to
Theorem
\ref{theorem: asymptotic concentration},
the limit distribution of velocity $\{p_t\}_{t = 1}^{\infty}$ defined in 
\ref{subsection: asymptotic behavior of distributions} 
should have compact support.
To see what happens, let us look at the following results of numerical calculations.

\begin{figure}
\includegraphics[width = 80mm]{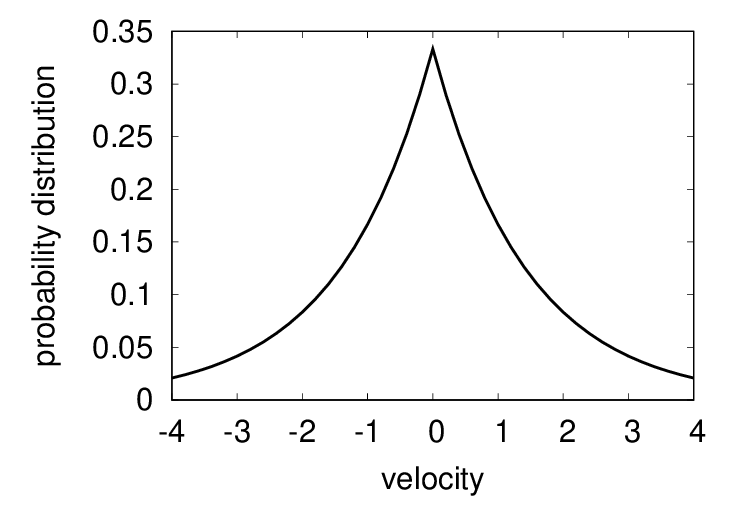}
\caption{The distribution of position according to the initial state $\xi$ in Subsection 4.3}
\end{figure}
\begin{figure}
\includegraphics[width = 80mm]{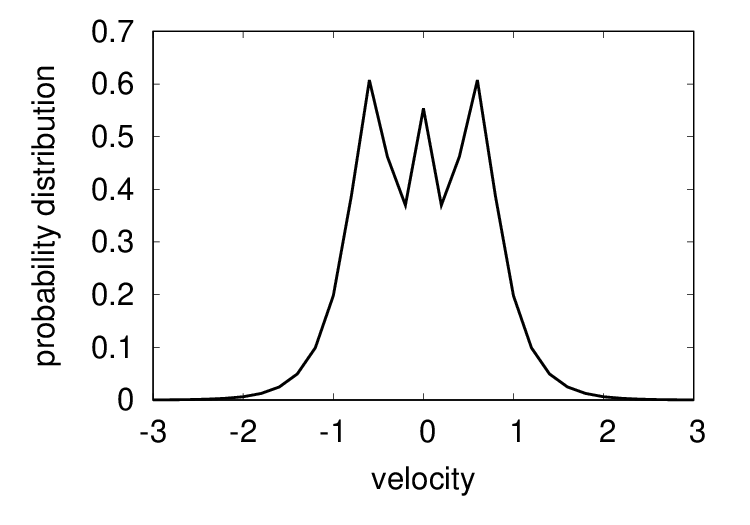}
\caption{The distribution of velocity $p_5$ given by the state $U^5 \xi$}
\end{figure}
\begin{figure}
\includegraphics[width = 80mm]{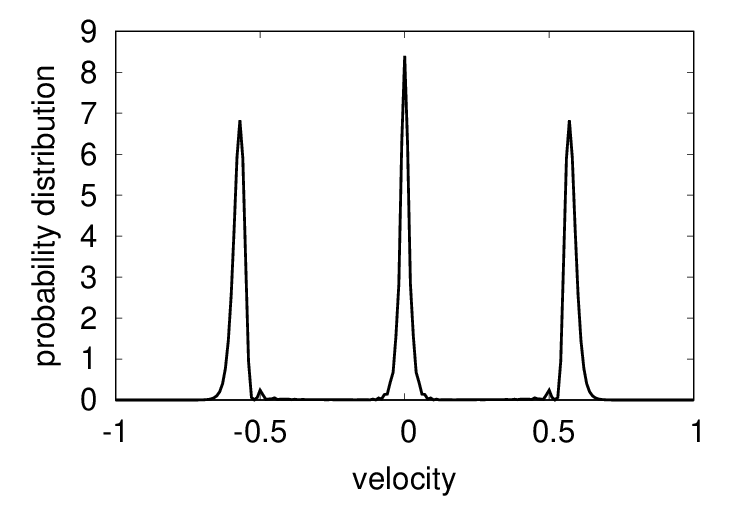}
\caption{The distribution of velocity $p_{100}$ given by the state $U^{100} \xi$}
\end{figure}

FIGURE 1 shows the probability of position in $\Z$ according to the quantum state 
$\xi$.
The left end stands for $x = -4$ and the right end stands for $x = 4$.
Since the support of $\xi \colon \Z \to \C^3$ is $\Z$, 
the probability never vanishes.

FIGURE 2 shows the probability distribution on $\frac{1}{5}\Z$ defined by the unit vector
$U^5 \xi$.
The left end stands for $x = -3$ and the right end stands for $x = 3$.
The tails on the both sides become closer the the level of $0$.

FIGURE 3 shows the probability distribution on $\frac{1}{100}\Z$ defined by the unit vector
$U^{100} \xi$.
The left end stands for $x = -1$ and the left end stands for $x = 1$.
Although the support of $U^{100} \xi \colon \Z \to \C^3$ is not supported on a compact set, the distribution is very close to that with compact support and the difference is invisible.
The probability on the interval $[- 0.7, 0.7]$ is more than $0.999$.

\section{Interpretation of this paper from the view point of quantum physics}

Many researchers have intensively studied quantum walks acting on the Hilbert space $\ell_2(\Z^d) \otimes \C^n$.
The natural number $d$ stands for the dimension of the space.
The natural number $n$ stands for the local degree of freedom at a point
in $\Z^d$.
This paper enables us to generalize the framework of the theory of quantum walks.
In the following, we see how our definitions and theorems in this paper enlarge the theory of quantum walks.

\paragraph{\bf General $n$}
In the present study of quantum walks, the local degree of freedom is restricted to small natural numbers $n$ such as $2$.
This is because researchers consider concrete physical observable such as the quantum spin. This observable can be described as a self-adjoint matrix acting on $\C^2$.
In our new framework described in Definition \ref{definition: QWs},
we can consider large finite dimensional local Hilbert space $\C^n$.
The results in Section \ref{section: general theory of asymptotic behavior}
can be applied to the case that the local Hilbert spaces are infinite dimensional.

\paragraph{\bf General $d$}
The space $\Z^d$ has a physical meaning in the case of $d = 1, 2, 3$.
We no longer have to divide our studies into these three cases.

\paragraph{\bf Quantum walks on more general spaces}
It is no longer necessary to stick to the concrete lattice $\Z^d$ inside $\R^3$.
Let us consider the case that some devices or atoms are located on some discrete subset $X \subset \R^d$. 
In this paragraph, we put no requirement related to symmetry on $X$. Let us consider some Hilbert spaces $\Hil_x$ is attached to every point $x \in X$, where $\Hil_x$ may be infinite dimensional.
Here, unit vectors in $\Hil_x$ describe quantum states at $x \in X$.
The whole Hilbert space $\Hil$ is defined by the direct sum Hilbert space $\oplus_{x \in X} \Hil_x$, and the spectral measure $E$ is given by the orthogonal projections $E(\Omega)$ from $\Hil$ onto $\oplus_{x \in \Omega \cap X} \Hil_x$.
Therefore, our new framework encompasses dynamical systems 
on arbitrary solid structures given by atoms and on those by devices.

\paragraph{\bf Quantum walks without finite propagation}
Almost all the known quantum walks has {\it finite propagation}.
The term ``{\it finite propagation}'' is defined as follows:
the operator $U$ on $\Hil = \oplus_{x \in X} \Hil_x$ is said to have finite propagation, if there exists a positive number $R$ such that for every $x, y \in X$, if $\mathrm{dist}(x, y) > R$, then $U \Hil_x$ and $\Hil_y$ are perpendicular.
The meaning of this condition is that if $x, y$ are distant, and $\xi \in \Hil$ is located at $\Hil_x$ and $\eta \in \Hil$ is located at $\Hil_y$, then the transition probability $|\langle U \xi, \eta \rangle|^2$ is zero.
It is easy to show that having finite propagation implies analyticity defined in Definition \ref{definition: analyticity}.
The converse does not hold.
The theory in this paper uses analyticity or more mild conditions.
Thus we obtain a wider framework.

\paragraph{\bf Quantum walks on arbitrary crystal lattices}
Preparing a general framework is not the only goal of this paper.
Let us apply our mathematical argument to quantum walks on crystal lattices.
For every crystal lattice $X \subset \R^d$, there exists an additive subgroup $G \subset \R^d$ which is isomorphic to $\Z^d$ such that $X$ is invariant under the addition by $G$.
Here $d$ is $2$ or $3$.
The Hilbert space $\Hil = \oplus_{x \in X} \Hil_x$ also has translation symmetry under the action of $G$,
in other words, for every $g \in G$ and $x \in X$, $g + x \in X$ and $\Hil_{g + x} = \Hil_x$.
Let $U$ be a unitary operator on $\Hil$.
The unitary operator satisfies a convergence theorem, if the following assumption holds:
\begin{itemize}
\item
All the local Hilbert spaces $\Hil_x$ are finite dimensional.
\item
The unitary operator $U$ has the {\it translation symmetry} with respect to $G$. More precisely, if we denote by $\rho(g)$ the unitary operator 
on $\Hil$ given by the shift $g \in G$, then $U \rho(g) = \rho(g) U$.
\item
The unitary operator $U$ has {\it finite propagation}.
\end{itemize}
Almost all the known quantum walks with translation symmetry on crystal lattices satisfy these conditions.

\begin{corollary}
For any initial unit vector $\xi \in \Hil$, and for the quantum walk $(\Hil, (U^t)_{t \in \Z}, E)$ on arbitrary crystal lattice satisfying above conditions,
the weak limit of the distribution of velocity $\{p_t\}$ defined in Subsection
\ref{subsection: asymptotic behavior of distributions} exists.
\end{corollary}

\begin{proof}
A quantum walk with {\it finite propagation} is {\it analytic}.
The quantum walk $U$ is space-homogeneous with respect to the additive group $G \cong \Z^d$.
This corollary is a direct conclusion of Theorem \ref{theorem: limit distribution}.
\end{proof}

%

\paragraph{\bf Physical meaning of $1$-cocycles and logarithmic derivatives of quantum walks}
Let $A$ be an observable described by a self-adjoint operator acting on a dense subspace of $\Hil = \oplus_{x \in X} \Hil_x$.
Let $(U^t)_{t \in \Z}$ be a dynamical system acting on $\Hil$.
According to the Heisenberg representation of time evolution,
the sequence $(U^{-t} A U^t)_{t \in \Z}$ stands for the evolution of the observable with time.
The $U^{-t} A U^t - A$ stands for the difference between two observables $A$ and $U^{-t} A U^t$, where the latter represents the observable after $t$.
Let us consider the case that $U^{-1} A U - A$ is bounded.
In this case the possible values of the observable $U^{-1} A U - A$ is restricted to some closed interval.
The family of bounded self-adjoint operators $(c_t)_{t \in \Z} = (U^{-t} A U^t - A)_{t \in \Z}$ 
forms a $1$-cocycle introduced in Definition \ref{definition: 1-cocycle}.
Thus using $1$-cocycle, we can treat the difference of the observable $A$ and that after $t$.
Proposition \ref{proposition: convergence of LD for homogeneous QW}
means that the average of the $1$-cocycle $\frac{1}{t} (U^{-t} A U^t - A)$
converge if the quantum walk is analytic and homogeneous.
This can be applied to analytic quantum walks on crystal lattices with translation symmetry.

In our mathematical argument for the convergence theorem (Theorem
\ref{theorem: limit distribution}),
a $1$-cocycle called the logarithmic derivative plays a key role.

\end{document}